\documentclass[12pt]{article}
\pdfoutput = 1
\usepackage[font=small,labelfont=bf]{caption}
\usepackage{amsmath, amsfonts,latexsym,amsthm, amsxtra, amssymb,graphicx,color}
\usepackage{multirow}
\usepackage{cite}
\usepackage{graphicx}
\usepackage[colorlinks,
            pdfborder=001,
            linkcolor=blue,
            anchorcolor=blue,
            citecolor=blue
            ]{hyperref}

\newtheoremstyle{mythm}{1.5ex plus 1ex minus .2ex}{1.5ex plus 1ex minus .2ex}
    {\it}{}{\bf}{}{0.5em}{}
\theoremstyle{mythm}

\newtheorem{thm}{Theorem}

\numberwithin{thm}{section}
\numberwithin{rem}{section}
\numberwithin{ex}{section}
\def\beq{\begin{equation}}
\def\eeq{\end{equation}}

\usepackage{titlesec}
\titleformat{\section}[hang]{\normalsize\bfseries}{\thetitle\quad}{0.0ex}{}
\titleformat{\subsection}[hang]{\sl\normalsize}{\thetitle\quad}{0.0ex}{}
\titlespacing*{\section} {0pt}{10pt}{4pt}

\usepackage[hang,flushmargin]{footmisc} %�ý�ע����������
\usepackage[authoryear]{natbib}
\begin{document}

%------------------------------------ title --------------------------------------------
\begin{center}{\Large\bf  Bayesian jackknife empirical likelihood with complex  surveys\footnotetext{Supported by}}
\vskip 4mm

{\bf Mengdong Shang, Xia Chen\footnote{Corresponding author. E-mail: xchen80@snnu.edu.cn}}

\small School of Mathematics and Statistics, Shaanxi Normal University, Xi'an 710119, China

\end{center}
\vskip 4mm

%----------------------------------- abstract ------------------
\noindent{\bf Abstract}\quad
We introduce a novel approach called the Bayesian Jackknife empirical likelihood method for analyzing survey data obtained from various unequal probability sampling designs.  This method is particularly applicable to parameters described by $U$-statistics.  Theoretical proofs establish that under a non-informative prior, the Bayesian Jackknife pseudo-empirical likelihood ratio statistic converges asymptotically to a normal distribution.  This statistic can be effectively employed to construct confidence intervals for complex survey samples.  In this paper, we investigate various scenarios, including the presence or absence of auxiliary information and the use of design weights or calibration weights.  We conduct numerical studies to assess the performance of the Bayesian Jackknife pseudo-empirical likelihood ratio confidence intervals, focusing on coverage probability and tail error rates.  Our findings demonstrate that the proposed methods outperform those based solely on the Jackknife pseudo-empirical likelihood, addressing its limitations.

\noindent{\bf Keywords}\quad
Bayesian inference; Jackknife empirical likelihood; Unequal probability sampling; Design weight; Calibration weight.

\section{Introduction}
With the advancement of society and the economy, complex surveys have emerged as a powerful data collection approach across various scientific domains. Notably, a significant development in recent years is the utilization of the empirical likelihood method in complex survey analysis, which offers the advantage of not relying on distributional assumptions. Pioneered by \citet{Owen88, Owen90}, the empirical likelihood was initially introduced for independent data, enabling the construction of confidence regions for nonparametric and semiparametric inference. This method retains two crucial properties of parametric likelihood, namely Wilks' theorem and Bartlett correction \citet{DiCiccio89, DiCiccio91}, facilitating the extension of empirical likelihood to numerous applications.

In the field of surveys, the notion of empirical likelihood was initially used by \citet{Hartley68} under the term "scale-load" method. Subsequently, \citet{Chen93} formally applied this approach for estimating the population means in the context of simple random sampling. Building upon this, \citet{Chen99} extended the application of pseudo-empirical likelihood to complex surveys with general unequal probability sampling designs, incorporating auxiliary information. They demonstrated that when auxiliary variables are utilized to estimate the population means, the proposed method asymptotically converges to the generalized regression estimation method. Expanding the scope, \cite{Wu06} presented pseudo-empirical likelihood ratio confidence intervals for a single parameter under arbitrary sampling designs. \cite{Zhao22} investigated the sample empirical likelihood methodology for point estimation and hypothesis testing of finite population parameters under a general unequal probability sampling design. They considered estimation equations involving smooth or non-differentiable functions, allowing for over-recognition, and utilized penalty sample empirical likelihood for variable selection, thereby establishing oracle properties.

To expand the application of empirical likelihood to nonlinear statistics, \cite{Jing09} introduced the jackknife empirical likelihood for estimating scalar parameters in one-sample and two-sample $U$-statistics. They demonstrated that the performance of this method surpasses that of the scalar empirical likelihood proposed by \cite{Qin06}. Additionally, \cite{Li16} extended the jackknife empirical likelihood method to encompass vector parameters and nonsmooth estimation equations. This extension allows for broader applicability and greater flexibility in statistical analysis.

Empirical likelihood holds the fundamental properties of estimators, namely consistency and asymptotic normality, enabling its utilization within Bayesian frameworks. The concept of Bayesian empirical likelihood was introduced by \citet{Lazar03} for independent and identically distributed samples. The approach involved setting a prior distribution on the unknown mean value of the generating distribution, which demonstrated the validity and frequency properties of the resulting posterior inference in the framework by Monahan \&\ Boos. It showed that the posterior distribution followed an asymptotic normal distribution. \citet{Cheng19} proposed the Bayesian jackknife empirical likelihood approach, which incorporates the parameter vector defined by U-statistics. They replaced the likelihood component in the Bayesian theorem with jackknife empirical likelihood. The authors established that the resulting posterior confidence interval maintains the correct coverage rate, and the posterior distribution approximately follows a normal distribution. Furthermore, \cite{Rao10} established a method for constructing asymptotically valid Bayesian pseudo-empirical likelihood intervals using the pseudo-empirical likelihood function in the context of general uni-stage unequal probability sampling designs. This approach proves to be highly flexible in incorporating auxiliary population information into the analysis. This method can be adapted to two practical situations: using the basic design weights to combine the known auxiliary general information into the construction of the interval; and using the calibration weights based on known auxiliary population means or totals. \cite{Zhao20} analyzed a finite population parameter vector defined by an over-recognizable estimation system with smooth and non-differentiable functions for complex survey data. In the framework of design-based, Bayesian point estimators and confidence intervals are obtained, and the large sample properties of Bayesian inference are established for no information prior and information prior. They also proposed a Bayesian model selection procedure and an MCMC program for calculating a posterior distribution. The finite sample performance as well as the influence of different types of prior and different types of sampling design, are verified by simulation studies.

This article focuses on the development of the Bayesian jackknife pseudo-empirical likelihood method for parameters defined by the $U$-statistic within the context of general uni-stage unequal probability sampling designs. Section 2 provides a detailed discussion of the Bayesian jackknife empirical likelihood method, considering both basic design weights and calibration weights. In Section 3, we analyze the asymptotic properties of Bayesian jackknife pseudo-empirical log-likelihood ratio functions, both with and without auxiliary population information. In Section 4, we present the results of simulation studies conducted to evaluate the performance of the proposed method. Moreover, we demonstrate the application of the method to the data set obtained from the 2019 China Household Finance Survey (CHFS) in Section 5. Finally, we conclude our findings in Section 6 and provide proof of the theoretical results in the Appendix.

\section{Methodology}
Let $\mathbf{y}= \left ( y_{1}, \ldots, y_{n} \right )^{T}$ represent a vector of independent and identically distributed observations of the random variable $Y$. And let $\theta$ denote a parameter of interest in the finite population. We define the $U$-statistic as follows:
\begin{equation*}
T_{n} = T \left ( y_{1}, \ldots, y_{n} \right ) = \begin{pmatrix} n \\ m \end{pmatrix}^{-1} \sum_{ 1 \leq i_{1} < \cdots < i_{m} \leq n } h \left ( y_{i_{1}}, \ldots, y_{i_{m}} \right ) 
\end{equation*}
be an unbiased estimator of $\theta$ and $h \left ( \cdot \right)$ be a kernel function. Define the jackknife pseudo-value,
\begin{equation*}
\hat{v}_{i}=nT_{n}-\left ( n-1 \right )T_{n-1}^{\left ( -i \right )}
\end{equation*}
can be used as a novel estimator for $\theta$ pertaining to the $i$th unit. Here, $T_{n-1}^{\left ( -i \right )}=T\left ( y_{1}, \ldots, y_{i-1}, \right.$ 
$\left. y_{i+1},\ldots, y_{n} \right )$ is obtained from the original data set using $n-1$ variables by omitting the $i$th data value.

Then, under the conditions $ \sum_{i = 1}^{n} p_{i} = 1$ and $ \sum_{i =1}^{n} p_{i} \left ( \hat{v}_{i} - \theta \right ) =0$, the estimation of $\theta$ via jackknife empirical likelihood involves the maximization of $\sum_{i =1}^{n} log \left ( p_{i} \right )$. To achieve this, we employ the method of Lagrange multipliers by setting the partial derivative of

\begin{equation*}
G = \sum_{i =1}^{n} log \left ( p_{i} \right ) - n \lambda \sum_{i =1}^{n} p_{i} \left ( \hat{v}_{i} - \theta \right ) - \gamma \left ( \sum_{i = 1}^{n} p_{i} - 1 \right )
\end{equation*}
with respect to $p_{i}$ equal to $0$. Then let
\begin{equation*}
\ell_{JEL} \left ( \theta \right ) = -n log \left ( n \right ) - \sum_{i =1}^{n} log \left \{ 1+ \lambda \left ( \theta \right ) \left ( \hat{v}_{i} - \theta \right ) \right \}
\end{equation*}
and $ \lambda = \lambda \left ( \theta \right )$ solves
\begin{equation*}
\sum_{i =1}^{n} \left ( \hat{v}_{i} - \theta \right ) / \left \{ 1+ \lambda \left ( \hat{v}_{i} - \theta \right ) \right \} = 0.
\end{equation*}
We integrate $L_{JEL} \left ( \theta \right ) = exp \left \{ \ell_{JEL} \left ( \theta \right ) \right \}$ with a predefined prior distribution $\pi \left ( \theta \right )$ for $ \theta $ using the Bayesian theorem to gain a pseudo posterior distribution
\begin{equation} \label{pi}
\pi \left ( \theta | \mathbf{ \hat{v}} \right) = c \left ( \mathbf{ \hat{v}} \right ) exp \left [ log \left \{  \pi \left ( \theta \right ) \right \} - \sum_{i = 1}^{n} log \left \{ 1+ \lambda \left ( \theta \right ) \left ( \hat{v}_{i} - \theta \right ) \right \} \right ],
\end{equation}
where $c \left ( \mathbf{ \hat{v}} \right )$ is the normalizing constant such that $ \int \pi \left ( \theta | \mathbf{ \hat{v}} \right ) d \theta =1 $. The $ \alpha $th quantile $t_{ \alpha}$ of $ \pi \left ( \theta | \mathbf{ \hat{v}} \right )$ is determined by $ H \left ( t_{ \alpha} | \mathbf{\hat{v}} \right )  = \alpha $ where
\begin{equation*}
H \left ( t_{ \alpha} | \mathbf{\hat{v}} \right ) = \int_{- \infty}^{t} \pi \left ( \theta | \mathbf{ \hat{v}} \right ) d \theta .
\end{equation*}
The validity of posterior inferences based on equation (\ref{pi}) can be assessed by examining the coverage probabilities of intervals $ \left ( t_{\alpha_{1}}, t_{\alpha_{2}} \right )$ for $\theta$. This assessment can be conducted through Monte Carlo simulations for selected values of $ \alpha_{1}$ and $ \alpha_{2}$. As proposed by \citet{Monahan92}, the posterior distribution $ \pi \left ( \theta | \mathbf{ \hat{v}} \right)$ should be considered valid if the statistic $H = H \left ( t_{ \alpha} | \mathbf{\hat{v}} \right )$ follows a uniform distribution over $ \left ( 0,1 \right )$.

Now assume a finite population comprising $N$ units, where each unit is characterized by the variable $y_{i}$ and a vector of auxiliary variables, denoted as $\mathbf{x}_i$. We take into account that the population mean vector is already known and represented as $\bar{\mathbf{X}} = N^{-1} \sum_{i=1}^{N} \mathbf{x}_i$. Let $s$ represent the set of $n$ units selected in the sample and for each unit, we define $\pi_i = P(i \in s), i = 1, \ldots, N$ as the probability of inclusion in the sample.

\subsection{Bayesian jackknife empirical likelihood with basic design weights}
Let us consider uni-stage sampling designs with a fixed sample size $n$. In this context, the basic design weights $d_i = 1/\pi_i$ and the normalized design weights $\tilde{d}_i(s) = d_i / \sum_{i \in s} d_i$ are employed for the given sample $s$.

The pseudo-empirical likelihood (PEL) method for complex survey data was initially proposed by \citet{Chen99}. \citet{Wu06} introduced the concept of the profile pseudo-empirical log-likelihood (PELL) function for estimating $\theta$. Building upon this, we extend the PELL framework to incorporate the jackknife pseudo-value, which is given by
\begin{equation} \label{JEL1}
\ell_{JEL_{d}} \left ( \theta \right ) = n^{\ast} \sum_{i \in s} \tilde{d}_{i} \left ( s \right ) \log \left \{ p_{i} \left ( \theta \right ) \right \}.
\end{equation}
The pseudo-modified empirical likelihood estimator (pseudo-MELE) for the parameter $\theta$ can be obtained using the Hajek estimator $\hat{V}_H = \sum_{i \in s} \hat{p}_{i} \hat{v}_i =\sum_{i \in s} \tilde{d}_i(s) \hat{v}_i$.

In the absence of auxiliary population information, we consider the effective sample size denoted as $n^* = n / deff_H$, where the design effect $deff_H$ is computed based on the Hajek estimator $\hat{V}_H$. The design effect is defined as
\begin{equation} \label{deffH}
deff_{H} = V_{p} \left ( \hat{V}_{H} \right ) / \left ( S_{v}^{2} / n \right ).
\end{equation}
The symbol $V_{p} \left ( \cdot \right )$ represents the variance under the specific design, while $S_{v}^2 / n$ denotes the variance of $\hat{V}_{H}$ under simple random sampling. When considering a fixed $\theta$, the objective is to maximize (\ref{JEL1}) while satisfying the constraints $p_i > 0$, $\sum_{i \in s} p_i = 1$, and
\begin{equation} \label{Con1}
\sum_{i \in s}p_{i}\hat{v_{i}}=\theta,
\end{equation}
leads to $\hat{p}_{i} \left ( \theta \right ) = \tilde{d}_{i} \left ( s \right ) / \left \{ 1 + \lambda \left ( \hat{v}_{i} - \bar{V} \right ) \right \}, \bar{V} = \sum_{i \in s} \hat{v}_{i} /n $, where the Lagrange multipiler $\lambda$ is the solution to 
\begin{equation} \label{Sol1}
\sum_{i \in s} \frac{\tilde{d}_{i} \left ( s \right ) \left ( \hat{v}_{i} - \bar{V} \right )}{1 + \lambda \left ( \hat{v}_{i} - \bar{V} \right )} = 0.
\end{equation}

Under the non-informative prior $\pi \left ( \theta \right ) \propto 1 $, the posterior density function of $\theta$ within the framework of the pseudo-Bayesian jackknife empirical likelihood is given by
\beq \label{Post1}
\pi_{1} \left ( \theta | \mathbf{\hat{v}} \right ) = c \left ( \mathbf{\hat{v}} \right ) exp \left [ -n^{\ast} \sum_{i \in s} \tilde{d}_{i} \left ( s \right ) log \left \{ 1 + \lambda \left ( \theta \right ) \left ( \hat{v}_{i} - \theta \right ) \right \} \right ],
\eeq
where $c \left ( \mathbf{\hat{v}} \right )$ is the normalizing constant such that $ \int \pi_{1} \left ( \theta | \mathbf{\hat{v}} \right ) d \theta = 1$.

We proceed to incorporate auxiliary population information by introducing an additional constraint, namely,
\begin{equation} \label{Con2}
\sum_{i \in s}p_{i}\mathbf{x}_{i}= \bar{\mathbf{X}}.
\end{equation}

According to \citet{Wu06}, the generalized regression (GREG) estimator of $\bar{V}$ is defined as $\hat{V}_{GR} = \hat{V}_{H} + {B}^\text{T} \left( \bar{\mathbf{X}} - \hat{\bar{\mathbf{X}}}_{H} \right)$, where $\hat{\bar{\mathbf{X}}}_{H} = \sum_{i \in s} \tilde{d}_{i} \left( s \right) \mathbf{x}_{i}$ represents the Hajek estimator, and
\begin{equation} \label{B}
B = \left \{ \frac{1}{N} \sum_{i=1}^{N} \left ( \mathbf{x}_{i} - \hat{\bar{\mathbf{X}}} \right ) {\left ( \mathbf{x}_{i} - \hat{\bar{\mathbf{X}}} \right ) }^\text{T} \right \}^{-1} \left \{ \frac{1}{N} \sum_{i=1}^{N} \left ( \mathbf{x}_{i} - \hat{\bar{\mathbf{X}}} \right )\left ( \hat{v} _{i} - \bar{V} \right ) \right \}
\end{equation}
is the vector of population regression coefficients. In this scenario, the design effect associated with $\hat{V}_{GR}$ differs from the one defined in (\ref{deffH}). It is defined as follows:
\begin{equation*}
deff_{GR} = V_{p} \left ( \hat{V}_{GR} \right ) / \left ( S_{r}^{2} / n \right ),
\end{equation*}
where $\hat{V}_{GR} = \sum_{i \in s} \tilde{d}_{i} \left ( s \right ) r_{i}$, $r_{i} = \hat{v}_{i} - \bar{V} - {B}^\text{T} \left ( \mathbf{x}_{i} - \bar{\mathbf{X}} \right )$, $S_{r}^{2} / n$ is the variance of $\hat{V}_{GR}$ under simple random sampling, and $S_{r}^{2} = \left ( N - 1 \right )^{-1} \sum_{i = 1}^{N} r_{i}^{2}$. Consequently, the effective sample size can be obtained as $n^{\ast}= n / deff_{GR}$.

Let $\mathbf{u}_{i} = {\left ( \hat{v}_{i} - \theta, {\left ( \mathbf{x}_{i} - \bar{\mathbf{X}} \right )}^\text{T} \right ) }^\text{T} $. It can be demonstrated that, under a fixed $\theta$, maximizing (\ref{JEL1}) subject to $p_{i}>0$, \ $\sum_{i \in s}p_{i}=1$, (\ref{Con1}) and (\ref{Con2}), leads to the expression $\tilde{p}_{i} \left ( \theta \right ) = \tilde{d}_{i} \left ( s \right ) / \left \{ 1 + {\bf{\lambda}}^\text{T} \left ( \theta \right ) \mathbf{u}_{i} \right \}$, where the vector-valued Lagrange multiplier $\bf{\lambda} = \bf{\lambda} \left ( \theta \right )$ is the solution to the following equation:
\begin{equation} \label{Sol2}
\sum_{i \in s} \frac{\tilde{d}_{i} \left ( s \right ) \mathbf{u}_{i}}{1 + {\bf{\lambda}}^\text{T} \mathbf{u}_{i}} = 0.
\end{equation}
The existence and uniqueness of a solution to (\ref{Sol2}) can be established if the zero vector is an interior point of the convex hull of the set $\left \{ \mathbf{u}_{i}, i \in s \right \}$, as shown in \citet{Chen02}.

Under the non-informative prior $\pi \left ( \theta \right ) \propto 1 $, the posterior density function of $\theta$ within the framework of the pseudo-Bayesian jackknife empirical likelihood is given by
\beq \label{Post2}
\pi_{2} \left ( \theta | \mathbf{\hat{v}}, \mathbf{u} \right ) = c \left ( \mathbf{\hat{v}}, \mathbf{u} \right ) exp \left [ - n^{\ast} \sum_{i \in s} \tilde{d}_{i} \left ( s \right ) log \left \{ 1 + \lambda' \left ( \theta \right ) \mathbf{u}_{i} \right \} \right ],
\eeq
where $c \left ( \mathbf{\hat{v}}, \mathbf{u} \right )$ is the normalizing constant such that $ \int \pi_{2} \left ( \theta | \mathbf{\hat{v}},\mathbf{u} \right ) d \theta = 1$.

\subsection{Bayesian jackknife empirical likelihood with calibration weights}
We now consider the calibration weights $w_{i}$, which are provided by statistical agencies. Let $\tilde{w}_{i} \left ( s \right ) = w_{i} / \sum_{i \in s} w_{i}$ be the normalized calibration weights, and $\sum_{i \in s} \tilde{w}_{i} \hat{v}_{i}$ is calibrated in the sense of $\sum_{i \in s} \tilde{w}_{i}\mathbf{x}_{i} = \mathbf{X}$, where $\mathbf{X} = N\bar{\mathbf{X}}$ is the known population total of the auxiliary vector $\mathbf{x}$.

The jackknife pseudo-empirical log-likelihood function for calibration weights is given by 
\begin{equation} \label{JEL2}
\ell_{JEL_{w}} \left ( \theta \right ) = m \sum_{i \in s} \tilde{w}_{i} \left ( s \right ) \log \left ( p_{i} \left ( \theta \right ) \right ),
\end{equation}
where $m$ is a scale factor to be determined. Here, we choose $m$ as a design-consistent estimator of $S_{v}^{2} / V_{p} \left ( \hat{V}_{GR} \right )$. For a fixed $\theta$, maximizing (\ref{JEL2}) subject to $p_{i} > 0$, $\sum_{i \in s} p_{i} = 1$, and constraint (\ref{Con1}) leads to $\breve{p}_{i} \left ( \theta \right ) = \tilde{w}_{i} \left ( s \right ) / \left \{ 1 + \lambda \left ( \hat{v}_{i} - \bar{V} \right ) \right \}$, where the Lagrange multiplier $\lambda$ satisfies
\begin{equation} \label{Sol3}
\sum_{i \in s} \frac{\tilde{w}_{i} \left ( s \right ) \left ( \hat{v}_{i} - \bar{V} \right )}{1 + \lambda \left ( \hat{v}_{i} - \bar{V} \right )} = 0.
\end{equation}
We observe that the equation (\ref{Con2}) does not restrict $\ell_{JEL_{w}}$ anymore.

Under the non-informative prior $\pi \left ( \theta \right ) \propto 1 $, the posterior density function of $\theta$ within the framework of the pseudo-Bayesian jackknife empirical likelihood is given by
\beq \label{Post3}
\pi_{3} \left ( \theta | \mathbf{\hat{v}} \right ) = c \left ( \mathbf{\hat{v}} \right ) exp \left [ - m \sum_{i \in s} \tilde{w}_{i} \left ( s \right ) log \left \{ 1 + \lambda \left ( \theta \right ) \left ( \hat{v}_{i} - \theta \right ) \right \} \right ],
\eeq
where $c \left ( \mathbf{\hat{v}} \right )$ is the normalizing constant such that $ \int \pi \left ( \theta | \mathbf{\hat{v}} \right ) d \theta = 1$.

\section{Main results}
In this section, we derive the asymptotic distribution of the Bayesian jackknife pseudo-empirical log-likelihood ratio functions, subject to specific regularity conditions, both when auxiliary population information is absent and present. We illustrate that, in all scenarios, the asymptotic distribution converges to a normal distribution. Elaborate proofs of these theorems are included in the appendix for reference.

\subsection{Bayesian jackknife empirical likelihood without auxiliary population information}
We first investigate the case of basic design weights without auxiliary population information. Maximizing the $\ell_{JEL_{d}}$ subject to $p_{i} > 0$ and $\sum_{i \in s}p_{i} = 1$ yields $\hat{p}_{i} = \tilde{d}_{i} \left ( s \right )$, and the effective sample size is denoted as $n^{\ast} = n / deff_{H}$.

To obtain our results, we need the following regularity conditions.

\noindent{\bf (C1)}\quad The sampling design $p\left ( s \right )$, the variable $y$, and the kernel function $h \left ( \cdot \right )$ satisfy
\begin{equation*}
\max_{i_{1}, \ldots, i_{m} \in s}\left | h \left( y_{i_{1}}, \ldots, y_{i_{m}} \right ) \right | = o_{p}\left ( n^{\frac{1}{2}} \right ),
\end{equation*}
where the random order $o_{p}\left ( \cdot  \right )$ is relevant for the sampling design $p\left ( s \right )$.

\noindent{\bf (C2)}\quad The sampling design $p\left ( s \right )$ satisfies $N^{-1} \sum_{i \in s}d_{i}-1=O_{p} \left (n^{ -1/2} \right)$.

Condition \noindent{\bf (C1)} imposes constraints on the sampling design $p \left ( s \right )$, the kernel function $h \left ( \cdot \right )$, and the finite population $\left \{ y_{1}, \ldots  , y_{N} \right \}$. These constraints imply that
\begin{equation*}
max_{i \in s} \left | nT_{n}-\left ( n-1 \right )T_{n-1}^{\left ( -i \right )} \right | = o_{p}\left ( n^{1/2} \right ),
\end{equation*}
which can be further simplified to $max_{i \in s} \left | \hat{v}_{i} \right | = o_{p}\left ( n^{1/2} \right )$.  Condition \noindent{\bf (C2)} states that $\hat{N} = \sum_{i \in s} d_{i}$ is a $\sqrt{n}$-consistent estimator of $N$.

%-------------------------------------------------------Theorem 1-----------------------------------------------------
\begin{thm}\label{thm1}
Under the conditions {\rm (C1)} and {\rm (C2)}, the posterior distribution of $\theta$ has the expansion for $ \{ \theta: \theta - \hat{\theta}_{JEL} = O_{p} \left ( 1/ \sqrt{n} \right ) \}:$
\begin{equation}
\pi_{1} \left ( \theta | \mathbf{\hat{v}} \right ) \propto exp \left \{ - \frac{1}{2} \left ( \theta - \hat{ \theta}_{JEL}  \right )^{2} /  V \left ( \hat{ \theta}_{JEL} \right ) + o_{p} \left ( 1 \right ) \right \},
\end{equation}
where $V \left ( \hat{ \theta}_{JEL} \right )$ is the design-based variance of $\hat{ \theta}_{JEL}$.

\end{thm}

\subsection{Bayesian jackknife empirical likelihood with auxiliary population information}
Now we consider the case of known auxiliary vectors $\mathbf{x}_{i}$ and basic design weights. Maximizing $\ell_{JEL_{d}}$ subject to $p_{i} > 0$, $\sum_{i \in s}p_{i} = 1$ and (\ref{Con2}), yields the estimated weights $\tilde{p}_{i}$ and an effective sample size $n^{\ast} = n / deff_{GR}$.

Assuming the additional regularity condition on the auxiliary vectors $\mathbf{x}_{i}$ hold,

\noindent{\bf (C3)}\quad $max_{i \in s}\left \| \mathbf{x}_{i}  \right \| = o_{p} \left ( n^{1/2} \right )$, where $\left \| \cdot \right \|$ represents the $L_{1}$ norm.

%-------------------------------------------------------Theorem 2-----------------------------------------------------
\begin{thm}\label{thm2}
Under the conditions {\rm (C1)-(C3)}, the posterior distribution of $\theta$ has the expansion for $ \{ \theta: \theta - \hat{\theta}_{PEL} = O_{p} \left ( 1/ \sqrt{n} \right ) \}:$
\begin{equation}
\pi_{2} \left ( \theta | \mathbf{\hat{v}}, \mathbf{x} \right ) \propto exp \left \{ - \frac{1}{2} \left ( \theta - \tilde{ \theta}_{JEL} \right )^{2} / V \left ( \tilde{ \theta}_{JEL} \right ) + o_{p} \left ( 1 \right ) \right \},
\end{equation}
where $V \left ( \tilde{ \theta}_{JEL} \right )$ is the design-based variance of $\tilde{ \theta}_{JEL}$.
\end{thm}

Next, we turn to the case of calibration weights with known auxiliary vectors $\mathbf{x}_{i}$. Under the same constraints, the expressions for $\breve{p}$ and $\breve{r}_{JEL_{w}}$ are identical to those of $\hat{p}$ and $\hat{r}_{JEL_{d}}$ except for the different weights. Therefore, a result similar to Theorem \ref{thm1} can be obtained. By maximizing $\ell_{JEL_{w}}$ subject to $p_{i} > 0$ and $\sum_{i \in s}p_{i} = 1$, we obtain $\breve{p}_{i}$ and set $m = S_{v}^{2} / V_{p} \left ( \hat{V}_{GR} \right )$.

We assume the following regularity condition.

\noindent{\bf (C4)}\quad The sampling design $p\left ( s \right )$ satisfies $N^{-1} \sum_{i \in s}w_{i}-1=O_{p} \left (n^{-1/2}\right) $.

%-------------------------------------------------------Theorem 3-----------------------------------------------------
\begin{thm}\label{thm3}
Under the conditions {\rm (C1)}, {\rm (C2)} and {\rm (C4)}, the posterior distribution of $\theta$ has the expansion for $ \left \{ \theta: \theta - \breve{\theta}_{PEL} = O_{p} \left ( 1/ \sqrt{n} \right ) \right \}:$
\begin{equation}
\pi_{3} \left ( \theta | \mathbf{\hat{v}} \right ) \propto exp \left \{ - \frac{1}{2} \left ( \theta - \breve { \theta}_{JEL} \right )^{2} / V \left ( \breve { \theta}_{JEL} \right ) + o_{p} \left ( 1 \right ) \right \},
\end{equation}
where $V \left ( \breve{ \theta}_{JEL} \right )$ is the design-based variance of $\breve{ \theta}_{JEL}$.
\end{thm}

\vskip 3mm

\section{Simulation studies}
In this section, we employ the Rao-Sampford method \citep{Rao65, Sampford67} for sampling without replacement, incorporating the Bayesian jackknife pseudo-empirical likelihood, to conduct simulation studies. Our goal is to assess the performance of various metrics, including coverage probabilities (CP), lower (L) and upper (U) tail error rates, average lengths of confidence intervals (AL), and average lower bound (LB) for the 95\% confidence intervals.

To generate finite populations, we adopt the following model \citep{Wu06}:
\begin{equation*}
y_{i}=\beta _{0}+\beta _{1}x_{i}+\sigma \varepsilon _{i},
\end{equation*}
where $\beta _{0}=\beta _{1}=1$, $x_{i}\sim \text{exp}(1)$, and $\varepsilon _{i}\sim \mathcal{N} (0,1)$. The population size is set as $N=1000$, and we consider two sample sizes: $n=100$ and $150$, corresponding to sampling fractions of 10\% and 15\%, respectively. To avoid extremely small values of $x_{i}$, a suitable constant number is added to all $x_{i}$. We examine two different values of $\sigma$ to reflect varying correlations between $y$ and $x$: $\rho (y,x)=0.3$ and $0.5$. The resulting finite populations remain consistent across repeated simulation runs. Our simulations are implemented in R, following the algorithms outlined in \citet{Wu04, Wu05} for accuracy and reliability.

We consider two examples: the probability-weighted moment and the estimate of variance. For the Bayesian jackknife pseudo-empirical likelihood method, we denote the interval based on no auxiliary variable as BJEL, the interval based on the basic design weights $d_{i}$ as $\rm{BJEL_{d}}$, and the interval based on the calibration weights $w_{i}$ as $\rm{BJEL_{w}}$. We also give corresponding results for intervals based on the jackknife empirical likelihood, which are denoted as JEL, $\rm{JEL_{d}}$ and $\rm{JEL_{w}}$. All the results are obtained from $B = 1000$ simulation runs.

\subsection*{Example 4.1: sample probability weighted moment}
Consider the probability weighted moment, $\theta = E \left \{ yF\left (  y\right ) \right \}$, where $F$ is the distribution function. The sample probability weighted moment is a $U$-statistic with the kernel $h= \max\left ( x,y \right )/2$.

Table \ref{tab1} presents the results of Example 4.1, which can be summarized as follows:

(a) The intervals of the Bayesian jackknife pseudo-empirical likelihood method perform better than those from the jackknife pseudo-empirical likelihood method, with coverage probabilities closer to the nominal value and more balanced tail error rates.

(b) The all intervals JEL, $JEL_d$, $JEL_w$, $BJEL_d$ and $BJEL_w$ have comparable average length.

\begin{table}[htbp]
\centering
\small 
\caption{Coverage probabilities, average lengths of the confidence intervals, and tail error rates for the probability-weighted moment in Example 4.1 are examined at a nominal level of 0.95.}\vskip -1mm
\label{tab1}
\begin{tabular}{cccccccc}
					\hline
					\makebox[0.1\textwidth][c]{$\rho$}  & \makebox[0.1\textwidth][c]{n} & \makebox[0.1\textwidth][c]{CI} & \makebox[0.1\textwidth][c]{CP(\%)} & \makebox[0.1\textwidth][c]{L} & \makebox[0.1\textwidth][c]{U} & \makebox[0.1\textwidth][c]{AL} & \makebox[0.1\textwidth][c]{LB} \\
					\hline
					0.3	&	100	&			 JEL		 	 & 93.1 & 2.7 & 4.2 & 0.627 & 3.611\\
							&			&			 BJEL		 	 & 94.0 & 2.8 & 3.2 & 0.639 & 3.623\\
							&			&   $JEL_{d}$  & 92.6 & 1.5 & 5.9 & 0.598 & 3.626\\
							&			&  $BJEL_{d}$  & 93.2 & 1.7 & 5.1 & 0.605 & 3.623\\
							&			&   $JEL_{w}$  & 92.6 & 1.5 & 5.9 & 0.598 & 3.619\\
							&			&  $BJEL_{w}$  & 93.1 & 2.0 & 4.9 & 0.608 & 3.630\\
							&	150	&			 JEL		 	 & 93.5 & 1.5 & 5.0 & 0.499  & 3.667\\
							&			&			 BJEL		 	 & 94.7 & 1.6 & 3.7 & 0.506 & 3.675\\
							&			&   $JEL_{d}$  & 92.4 & 1.3 & 6.3 & 0.480 & 3.675\\
							&			&  $BJEL_{d}$  & 93.3 & 1.4 & 5.3 & 0.484 & 3.682\\
							&			&   $JEL_{w}$  & 92.6 & 1.3 & 6.1 & 0.479 & 3.676\\
							&			&  $BJEL_{w}$  & 93.4 & 1.4 & 5.2 & 0.485 & 3.684\\
					0.5	&	100	&			 JEL		 	 & 93.6 & 2.2 & 4.2 & 0.375 & 3.368\\
							&			&			 BJEL		 	 & 94.2 & 2.7 & 3.1 & 0.385 & 3.376\\
							&			&   $JEL_{d}$  & 91.9 & 1.4 & 6.7 & 0.319 & 3.382\\
							&			&  $BJEL_{d}$  & 92.0 & 2.0 & 6.0 & 0.323 & 3.387\\
							&			&   $JEL_{w}$  & 91.5 & 1.7 & 6.8 & 0.321 & 3.386\\
							&			&  $BJEL_{w}$  & 92.3 & 2.3 & 5.4 & 0.327 & 3.392\\
							&	150	&			 JEL		 	 & 94.7 & 1.3 & 4.0 & 0.296 & 3.399\\
							&			&			 BJEL		 	 & 95.3 & 1.5 & 3.2 & 0.302 & 3.405\\
							&			&   $JEL_{d}$  & 91.9 & 1.3 & 6.8 & 0.256 & 3.415\\
							&			&  $BJEL_{d}$  & 92.2 & 1.4 & 6.4 & 0.258 & 3.418\\
							&			&   $JEL_{w}$	 & 92.1 & 1.3 & 6.6 & 0.257 & 3.417\\
							&			&  $BJEL_{w}$  & 93.1 & 1.5 & 5.4 & 0.260 & 3.421\\
					\hline
\end{tabular}
\end{table}

\subsection*{Example 4.2: sample variance}
Consider the sample variance, where the kernel $h = \left ( x-y \right )^{2}/2$. The $U$-statistic can be expressed as $U = \sum_{i=1}^{n} \sum_{j=1}^{n} \left ( y_{i} - y_{j} \right )^{2} / \left \{ n \left ( n-1 \right ) \right \}$.

Table \ref{tab2} presents the summary of results for Example 4.2. The table indicates that the coverage probabilities for the BJEL, $\rm{BJEL_{d}}$, and $\rm{BJEL_{w}}$ methods are all near the nominal values compared to the corresponding jackknife empirical likelihood, despite a small to moderate expansion in interval lengths. Furthermore, the Bayesian jackknife empirical likelihood method also partially mitigates the issue of imbalanced tail error rates that can arise from the pseudo-jackknife empirical likelihood method. It is worth noting that in the case of large $n$, the Bayesian jackknife empirical likelihood may be overfitting.

\begin{table}[htbp]
\centering
\small 
\caption{Coverage probabilities, average lengths of the confidence intervals, and tail error rates for the variance in Example 4.2 are examined at a nominal level of 0.95.}\vskip -1mm
\label{tab2}
\begin{tabular}{cccccccc}
					\hline
					\makebox[0.1\textwidth][c]{$\rho$}  & \makebox[0.1\textwidth][c]{n} & \makebox[0.1\textwidth][c]{CI} & \makebox[0.1\textwidth][c]{CP(\%)} & \makebox[0.1\textwidth][c]{L} & \makebox[0.1\textwidth][c]{U} & \makebox[0.1\textwidth][c]{AL} & \makebox[0.1\textwidth][c]{LB}\\
					\hline
					0.3	&	100	&			 JEL		 	 & 93.4 & 1.0 & 5.6 & 5.887 & 8.242\\
							&			&			 BJEL		 	 & 94.8 & 1.3 & 3.9 & 6.136 & 8.414\\
							&			&   $JEL_{d}$  & 93.4 & 1.0 & 5.6 & 5.820 & 8.236\\
							&			&  $BJEL_{d}$  & 94.0 & 1.6 & 4.4 & 5.984 & 8.402\\
							&			&   $JEL_{w}$  & 93.7 & 1.0 & 5.3 & 5.857 & 8.255\\
							&			&  $BJEL_{w}$  & 94.4 & 1.6 & 4.0 & 6.101 & 8.426\\
							&	150	&			 JEL		 	 & 94.8 & 0.5 & 4.7 & 4.663 & 8.788\\
							&			&			 BJEL		 	 & 94.9 & 0.7 & 4.4 & 4.813 & 8.907\\
							&			&   $JEL_{d}$  & 94.3 & 0.6 & 5.1 & 4.658 & 8.795\\
							&			&  $BJEL_{d}$  & 94.4 & 1.1 & 4.5 & 4.766 & 8.912\\
							&			&   $JEL_{w}$  & 94.3 & 0.7 & 5.0 & 4.662 & 8.809\\
							&			&  $BJEL_{w}$  & 94.8 & 1.1 & 4.1 & 4.810 & 8.929\\
					0.5	&	100	&			 JEL		 	 & 93.6 & 1.9 & 4.5 & 2.220 & 3.019\\
							&			&			 BJEL		 	 & 93.8 & 3.1 & 3.1 & 2.357 & 3.089\\
							&			&   $JEL_{d}$  & 93.2 & 1.8 & 5.0 & 2.336 & 3.017\\
							&			&  $BJEL_{d}$  & 94.2 & 2.5 & 3.3 & 2.179 & 3.079\\
							&			&   $JEL_{w}$  & 93.2 & 2.2 & 4.6 & 2.183 & 3.036\\
							&			&  $BJEL_{w}$  & 94.6 & 3.0 & 2.4 & 2.313 & 3.103\\
							&	150	&			 JEL		 	 & 95.2 & 1.0 & 3.8 & 1.743 & 3.218\\
							&			&			 BJEL		 	 & 95.2 & 1.9 & 2.9 & 1.832 & 3.266\\
							&			&   $JEL_{d}$  & 94.9 & 0.9 & 4.2 & 1.781 & 3.225\\
							&			&  $BJEL_{d}$  & 95.8 & 1.2 & 3.0 & 1.742 & 3.269\\
							&			&   $JEL_{w}$	 & 95.5 & 1.0 & 3.5 & 1.733 & 3.239\\
							&			&  $BJEL_{w}$  & 95.8 & 1.6 & 2.6 & 1.818 & 3.288\\
					\hline
\end{tabular}
\end{table}

\section{Application to the China Household Finance Survey data}
The China Household Finance Survey (CHFS) is an extensive nationwide sample survey project designed to collect comprehensive and meticulous information on household finance at the micro level. It aims to provide a detailed description of household economic and financial behaviors. The data utilized in this section are derived from the China Household Finance Survey project (CHFS), managed by the Survey and Research Center for China Household Finance at the Southwestern University of Finance and Economics.

We employ the proposed jackknife pseudo-empirical likelihood method on the 2019 CHFS dataset, which encompasses 29 provinces, 170 cities, 345 districts and counties, and 1,360 villages (residential) committees in China. The sample size consists of 34,643 households, rendering the data representative at both the national and provincial levels. It is important to note that the CHFS datasets are not publicly accessible but can be obtained through an approval process outlined at (\href{https://chfs.swufe.edu.cn/}{https://chfs.swufe.edu.cn/}).

The variable of interest in this study is resident happiness, denoted as Y. The dataset includes a column of household sample weight ($w_i$) for analytical purposes. We consider the sample to be drawn through a single-stage unequal probability sampling design, with the design weight ($d_i = \pi_i^{-1}$) equivalent to the calibration weight ($w_i$). The auxiliary population information ($x_i$) is derived from the first, second, and third-tier city division variables.

We focus on investigating the variance of resident happiness, represented by the $U$-statistic
\begin{equation*}
U = \sum_{i=1}^{n} \sum_{j=1}^{n} \left ( y_{i} - y_{j} \right )^{2} / \{ n \left ( n-1 \right ) \},
\end{equation*}
where $n=34,643$. Here we also consider non-information prior. The 95\% credible interval for the Bayesian jackknife empirical likelihood, without auxiliary population information, is estimated to be (0.762,0.785). For comparison, the confidence intervals for the Bayesian jackknife empirical likelihood with basic design weights and Bayesian jackknife empirical likelihood with calibration weights are both (0.759,0.782). Based on these findings, we can conclude that the city's level has a relatively minimal impact on resident happiness.

\section{Conclusion}
This article aims to demonstrate the potential of employing the Bayesian jackknife pseudo-empirical likelihood approach to address the estimation challenges associated with $U$-statistics in complex surveys. The estimation process takes into account the characteristics of the sampling design and the utilization of auxiliary information. The efficacy of this approach is assessed through both theoretical justification and empirical evidence. Compared to the jackknife pseudo-empirical likelihood interval, the Bayesian jackknife empirical likelihood interval exhibits superior performance in terms of balanced tail error rates and coverage probabilities.

We are currently exploring the extension of our proposed Bayesian jackknife empirical likelihood method to other survey sampling techniques, such as information prior. Recent years have witnessed significant advancements in non-probability sampling methods. It is intriguing to investigate how to capitalize on the strengths of both probability and non-probability sampling and combine them to achieve more accurate statistical estimates. For combined samples of this nature, employing the method proposed in this study holds great promise for similar analyses.

\section*{Acknowledgments}
This work was supported by the National Natural Science Foundation of China (12031016).

\appendix
\section{Proofs}

In this section, we provide rigorous mathematical proofs of the key findings.

%----------------------------------------The proof of theorem 1---------------------------
\begin{proof}[Proof of Theorem \rm\ref{thm1}]

By rephrasing $\tilde{d}_{i}\left ( s \right )\left ( \hat{v}_{i}-\bar{V} \right )$ as 
\begin{equation*}
\tilde{d}_{i}\left ( s \right )\left ( \hat{v}_{i}-\bar{V} \right )\left [ 1+\lambda\left ( \hat{v}_{i}-\bar{V} \right )-\lambda\left ( \hat{v}_{i}-\bar{V} \right ) \right ],
\end{equation*}
we can reorganize (\ref{Sol1}) to derive
\begin{equation} \label{A.1}
\lambda \sum_{i \in s}\frac{\tilde{d}_{i}\left ( s \right )\left ( \hat{v}_{i}-\bar{V} \right )^2}{1+\lambda\left ( \hat{v}_{i}-\bar{V} \right )}= \sum_{i \in s}\tilde{d}_{i}\left ( s \right ) \hat{v}_{i}-\bar{V}.
\end{equation}
It follows from (\ref{A.1}) that 
\begin{equation} \label{A.2}
\frac{\left |\lambda  \right |}{1+\left |\lambda  \right |u^{\ast }}\sum_{i \in s} \tilde{d}_{i}\left ( s \right )\left ( \hat{v}_{i}-\bar{V} \right )^2\leq \left | \sum_{i \in s} \tilde{d}_{i}\left ( s \right )\hat{v}_{i}-\bar{V} \right |,
\end{equation}
where $u^{\ast }=max_{i \in s}\left | \hat{v}_{i}-\bar{V} \right |$, which is of order $o_{p} \left (n^{1/2} \right )$ by Condition \noindent{\bf (C1)}. It follows from \citet{Hajek60, Hajek64}, we can get the central limit theorem for a Horvitz-Thompson estimator, namely $\hat{\bar{V}}_{HT}=N^{-1} \sum_{i \in s} d_{i} \hat{v}_{i}$ is asymptotically normally distributed. Then, we can obtain $\hat{\bar{V}}_{HT}=\bar{V}+O_{p}\left( n^{-1/2} \right )$.

Under Condition \noindent{\bf (C2)}, we have $\hat{N}/N=1+O_{p}\left( n^{-1/2} \right )$, where $\hat{N}=\sum_{i \in s} d_{i}$, which imply
\begin{equation*}
\sum_{i \in s} \tilde{d}_{i}\left ( s \right )\hat{v}_{i}=\hat{\bar{V}}_{HT}/\left ( \hat{N}/N \right )= \bar{V}+O_{p}\left ( n^{-1/2} \right ).
\end{equation*}
Noting that $\sum_{i \in s} \tilde{d}_{i}\left ( s \right ) \left( \hat{v}_{i}-\bar{V} \right )^2$ is the Hajek-type estimator of $S_{v}^{2}$, which is of order $O \left ( 1 \right )$. It follows from (\ref{A.2}) that we must have $\lambda = O_{p}\left ( n^{-1/2} \right )$ and, consequently, $max_{i \in s} \left | \lambda \left ( \hat{v}_{i}-\bar{V} \right ) \right |= o_{p}\left (1 \right )$. This together with (\ref{A.1}) leads to 
\begin{equation*}
\lambda =\left \{ \sum_{i \in s} \tilde{d}_{i}\left ( s \right ) \left ( \hat{v_{i}}-\bar{V} \right )^2 \right \}^{-1}\left ( \sum_{i \in s} \tilde{d}_{i}\left ( s \right )\hat{v_{i}}-\bar{V} \right )+o_{p}\left ( n^{-1/2} \right ).
\end{equation*}
Using a Taylor series expansion of $log \left( 1+x \right )$ at $x=\lambda \left ( \hat{v}_{i}-\bar{V} \right)$ up to the second order, we obtain 
\beq
2 \ell_{JEL_{d}} \left ( \theta \right ) - 2\ell_{JEL_{d}} \left ( \hat{\theta}_{JEL} \right ) = - n^{\ast } \left ( \theta - \hat{\theta}_{JEL} \right )^{2} / \left ( \sum_{i \in s} \tilde{d}_{i}\left ( s \right ) \left (\hat{v}_{i}- \theta \right )^2 \right )+o_{p}\left ( 1 \right ). \nonumber
\eeq
Using the above argument, we show that
\beq
\ell_{JEL_{d}} \left ( \theta \right ) = \ell_{JEL_{d}} \left ( \hat{\theta}_{JEL} \right ) - \frac{1}{2} \left ( \theta - \hat{\theta}_{JEL} \right )^{2} / V \left ( \hat{ \theta}_{JEL} \right ) + o_{p} \left ( 1 \right ) \nonumber
\eeq
for $ \theta = \left ( \hat{\theta}_{JEL} \right ) + O_{p} \left ( n^{-1/2} \right )$ and $V \left ( \hat{\theta}_{JEL} \right ) = n^{\ast } / \left ( \sum_{i \in s} \tilde{d}_{i}\left ( s \right ) \left (\hat{v}_{i}- \theta \right )^2 \right )$. Consequently, we have 
\beq
log \pi_{1} \left ( \theta | \mathbf{\hat{v}} \right ) = \frac{1}{2} \left ( \theta - \hat{\theta}_{JEL} \right )^{2} / V \left ( \hat{ \theta}_{JEL} \right ) + o_{p} \left ( 1 \right ). \nonumber
\eeq
\end{proof}

%-------------------------------The proof of theorem 2------------------------------------------
\begin{proof}[\bf\it Proof of Theorem \rm\ref{thm2}]

The discussions concerning the order of magnitude and the asymptotic expansion of the relevant Lagrange multiplier follow a similar approach to that presented in the proof of Theorem \ref{thm1}. Nevertheless, this proof incorporates two essential and distinct arguments. The $\tilde{p}_{i}$ which maximize $\ell_{JEL_{d}} \left( \theta \right )$ subject to $\sum_{i \in s}p_{i}=1$ and $\sum_{i \in s}p_{i}x_{i}=\bar{X}$ are given by $\tilde{p}_{i}=\tilde{d}_{i} \left ( s \right )/ \left \{ 1+\lambda^{T} \left ( x_{i}-\bar{X} \right ) \right \}$, where the $\lambda$ is the solution to 
\begin{equation*}
\sum_{i \in s} \frac{\tilde{d}_{i} \left ( s \right ) \left ( \bf{x}_{i} - \bar{\bf{X}} \right ) }{1 + {\lambda}' \left ( \bf{x}_{i} - \bar{\bf{X}} \right )} = 0.
\end{equation*}
Under Conditions \noindent{\bf (C2), (C3)} and $\hat{\bar{X}}_{HT}=N^{-1} \sum_{i \in s} d_{i} x_{i}$ is asymptotically normally distributed, we can show that $\left \| \lambda \right \|=O_{p} \left ( n^{-1/2} \right )$ and 
\begin{equation*}
\lambda = \left \{ \sum_{i \in s} \tilde{d}_{i} \left ( s \right ) \left ({x_{i}-\bar{X}} \right ) \left ( {x_{i}-\bar{X}} \right )^{T} \right \}^{-1} \left ( \sum_{i \in s} \tilde{d}_{i} \left( s \right ) x_{i}-\bar{X} \right ) + o_{p} \left ( n^{-1/2} \right ).
\end{equation*}
By omitting the term $n \sum_{i \in s} \tilde{d}_{i} \left ( s \right ) log \left ( \tilde{d}_{i} \left ( s \right ) \right )$, we obtain the following asymptotic expansion for $\ell_{JEL_{d}} \left ( \tilde{p} \right )$: 
\begin{equation} \label{A.3}
-\frac{n}{2} \left ( \sum_{i \in s} \tilde{d}_{i} x_{i} - \bar {X} \right )^{T}\left \{ \sum_{i \in s} \tilde{d}_{i} \left ( s \right )\left ({x_{i}-\bar{X}} \right )   \left ( {x_{i}-\bar{X}} \right )^{T} \right \}^{-1}\left ( \sum_{i \in s} \tilde{d}_{i} x_{i}- \bar {X} \right )+o_{p} \left ( 1 \right ).
\end{equation}

To obtain a similar expansion for $\ell_{JEL_{d}} \left ( \tilde{p} \left ( \bar{V} \right ) \right )$ where $\tilde{p} \left ( \bar{V} \right )$ maximize $\ell_{JEL_{d}} \left ( p \right )$ subject to $\sum_{i \in s}p_{i}=1$, (\ref{Con1}) and (\ref{Con2}), our first crucial argument is to reformulate the constrained maximization problem as follows: let $r_{i}=\hat{v}_{i}- \bar{V}-B^{T} \left ( x_{i}-\bar{X} \right )$ where $B$ is defined by (\ref{B}). Then the set of constraints is equivalent to
\begin{equation} \label{A.4}
\sum_{i \in s}p_{i}=1,\quad \sum_{i \in s}p_{i}x_{i}=\bar{X} \quad and \quad \sum_{i \in s}p_{i}r_{i} = 0.
\end{equation}
With complete parallel development that leads to $\ell_{JEL_{d}}$ given by (\ref{A.3}), maximizing $\ell_{JEL_{d}}\left ( \tilde{p} \right )$ subject to (\ref{A.4}) leads to the following expansion for $\ell_{JEL_{d}} \left ( \tilde{p} \left ( \bar{V} \right ) \right )$
\begin{equation} \label{A.5}
-\frac{n}{2} \left ( \sum_{i \in s} \tilde{d}_{i} \mathbf{u}_{i} - \bar {U} \right )^{T}\left \{ \sum_{i \in s} \tilde{d}_{i} \left ( s \right ) \left ( \mathbf{u}_{i} - \bar {U} \right )  \left ( \mathbf{u}_{i} -\bar {U} \right )^{T} \right \}^{-1}\left ( \sum_{i \in s} \tilde{d}_{i} \mathbf{u}_{i} - \bar {U} \right )+o_{p} \left ( 1 \right ),
\end{equation}
where $\mathbf{u}_{i}=\left ( x_{i}^{T}, r_{i} \right )^{T}$ and $\bar{U}= \left (\bar{X}^{T},0 \right )^{T}$. Our second crucial argument is the observation that $\sum^{N}_{i=1} \left (x_{i}-\bar{X} \right )r_{i}$, i.e., the matrix involved in the middle of (\ref{A.5}) is an estimate for its population counterpart which is block diagonal. It is straightforward to show that
\beq
2 \ell_{JEL_{d}} \left ( \theta \right ) - 2 \ell_{JEL_{d}} \left ( \tilde{\theta}_{JEL} \right ) = - n^{\ast } \left ( \theta - \tilde{\theta}_{JEL} \right ) ^{2} / V \left ( \tilde{ \theta}_{JEL} \right ) + o_{p} \left ( 1 \right ). \nonumber
\eeq

Using the above argument, we show that
\beq
\ell_{JEL_{d}} \left ( \theta \right ) = \ell_{JEL_{d}} \left ( \tilde{\theta}_{JEL} \right ) - \frac{1} {2} \left ( \theta - \tilde{\theta}_{JEL} \right )^{2} / V \left ( \tilde{ \theta}_{JEL} \right ) + o_{p} \left ( 1 \right ) \nonumber
\eeq
for $ \theta = \tilde{\theta}_{JEL} + O_{p} \left ( n^{-1/2} \right )$. Consequently, we have
\beq
log \pi_{2} \left ( \theta | \mathbf{\hat{v}}, \mathbf{x} \right ) = - \frac{1} {2} \left ( \theta - \tilde{\theta}_{JEL} \right )^{2} / V \left ( \tilde{ \theta}_{JEL} \right ) + o_{p} \left ( 1 \right ). \nonumber
\eeq
\end{proof}

%----------------------------------------The proof of theorem 3---------------------------
\begin{proof}[Proof of Theorem \rm\ref{thm3}]

This proof bears a resemblance to that of Theorem \ref{thm1}, differing solely in the incorporation of distinct design weighted. By rephrasing $\tilde{w}_{i}\left ( s \right )\left ( v_{i}-\bar{V} \right )$ as 
\begin{equation*}
\tilde{w}_{i}\left ( s \right )\left ( v_{i}-\bar{V} \right )\left [ 1+\lambda\left ( \hat{v}_{i}-\bar{V} \right )-\lambda\left ( \hat{v}_{i}-\bar{V} \right ) \right ],
\end{equation*}
we can recognize (\ref{Sol3}) to derive
\begin{equation} \label{A.6}
\lambda \sum_{i \in s}\frac{\tilde{w}_{i}\left ( s \right )\left ( v_{i}-\bar{V} \right )^2}{1+\lambda\left ( v_{i}-\bar{V} \right )}= \sum_{i \in s} \tilde{w}_{i} \left ( s \right ) \hat{v}_{i}- \bar{V}.
\end{equation}
It follows from (\ref{A.6}) that 
\begin{equation} \label{A.7}
\frac{\left |\lambda  \right |}{1+\left |\lambda  \right |u^{\ast }}\sum_{i \in s} \tilde{w}_{i}\left ( s \right )\left ( \hat{v}_{i}-\bar{V} \right )^2\leq \left | \sum_{i \in s} \tilde{w}_{i}\left ( s \right )\hat{v}_{i}-\bar{V} \right |,
\end{equation}
where $u^{\ast }=max_{i \in s}\left | \hat{v}_{i}-\bar{V} \right |$, which is of order $o_{p} \left (n^{1/2} \right )$ by Condition \noindent{\bf (C1)}. It follows from \citet{Hajek60, Hajek64}, we can get the central limit theorem for a Horvitz-Thompson estimator, namely $\breve{\bar{V}}_{HT}=N^{-1} \sum_{i \in s} w_{i} \hat{v}_{i}$ of $\bar{V}$ is asymptotically normally distributed. Then, we can require $\breve{\bar{V}}_{HT}=\bar{V}+O_{p}\left( n^{-1/2} \right )$.

Under Condition \noindent{\bf (C4)}, we have $\breve{N}/N=1+O_{p}\left( n^{-1/2} \right )$, where $\breve{N}=\sum_{i \in s} w_{i}$, which imply 
\begin{equation*}
\sum_{i \in s} \tilde{w}_{i}\left ( s \right )\hat{v}_{i}=\breve{\bar{V}}_{HT}/ \left ( \breve{N}/N \right )= \bar{V}+O_{p}\left ( n^{-1/2} \right ).
\end{equation*}
Noting that $\sum_{i \in s} \tilde{w}_{i}\left ( s \right ) \left( \hat{v}_{i}-\bar{V} \right )^2$ is the Hajek-type estimator of $S_{v}^{2}$ which is of order $O \left ( 1 \right )$, it follows from (\ref{A.2}) that we must have $\lambda = O_{p}\left ( n^{-1/2} \right )$ and, consequently, $max_{i \in s} \left | \lambda \left ( \hat{v}_{i}-\bar{V} \right ) \right |= o_{p}\left (1 \right )$. This together with (\ref{A.6}) leads to 
\begin{equation*}
\lambda =\left \{ \sum_{i \in s} \tilde{w}_{i}\left ( s \right ) \left ( \hat{v_{i}}-\bar{V} \right )^2 \right \}^{-1}\left ( \sum_{i \in s} \tilde{w}_{i}\left ( s \right )\hat{v_{i}}-\bar{V} \right )+o_{p}\left ( n^{-1/2} \right ).
\end{equation*}
Using a Taylor series expansion of $log \left( 1+x \right )$ at $x=\lambda \left ( \hat{v}_{i}-\bar{V} \right)$ up to the second order, we obtain
\beq
2 \ell_{JEL_{w}} \left ( \theta \right ) - 2 \ell_{JEL_{w}} \left ( \breve{\theta}_{JEL} \right ) = - m \left ( \theta - \breve{\theta}_{JEL} \right )^{2} / \left ( \sum_{i \in s} \tilde{w}_{i}\left ( s \right ) \left (\hat{v}_{i}- \theta \right )^2 \right ) + o_{p} \left ( \frac{m}{n} \right ). \nonumber
\eeq
Using the above argument, we show that
\beq
\ell_{JEL_{w}} \left ( \theta \right ) = \ell_{JEL_{w}} \left ( \breve {\theta}_{JEL} \right ) - \frac{1}{2} \left ( \theta - \breve{\theta}_{JEL} \right )^{2} / V \left ( \breve{ \theta}_{JEL} \right ) + o_{p} \left ( \frac{m}{n} \right ) \nonumber
\eeq
for $ \theta = \left ( \breve{\theta}_{JEL} \right ) + O_{p} \left ( n^{-1/2} \right )$ and $V \left ( \hat{\theta}_{JEL} \right ) = \frac{1}{m} \sum_{i \in s} \tilde{w}_{i}\left ( s \right ) \left (\hat{v}_{i}- \theta \right )^2 $. Consequently, we have
\beq
log \pi_{3} \left ( \theta | \mathbf{\hat{v}} \right ) =  - \frac{1}{2} \left ( \theta - \breve{\theta}_{JEL} \right )^{2} / V \left ( \breve{ \theta}_{JEL} \right ) + o_{p} \left ( \frac{m}{n} \right ). \nonumber
\eeq
\end{proof}
\bibliographystyle{cas-model2-names}
% Loading bibliography database
\bibliography{ref}

\end{document}